\documentclass[journal]{IEEEtran}

\ifCLASSINFOpdf
\else
   \usepackage[dvips]{graphicx}
\fi

\usepackage{amsmath,amssymb,amsfonts}
\usepackage{amsthm}
\usepackage{algorithm}
\usepackage{algpseudocode}
\usepackage{bbm}
\usepackage{array}
\usepackage{textcomp}
\usepackage{stfloats}
\usepackage{url}
\usepackage{subfigure}
\usepackage{mathtools}
\usepackage{dsfont}
\usepackage{verbatim}
\usepackage{graphicx}
\usepackage{cite,color}
\usepackage{multirow}

\usepackage[all=normal,paragraphs=tight,floats=normal,mathspacing=normal,wordspacing=normal,charwidths=tight,mathdisplays=normal,leading=tight]{savetrees}
\usepackage{balance}

\newtheorem{lemma}{Lemma}
\newtheorem{corollary}{Corollary}
\newtheorem*{remark}{Remark}
\newtheorem{theorem}{Theorem}
\begin{document}

\title{Average and Worst-case Analysis of MIMO Beamforming Loss due to Hardware Impairments}

\author{Xuan Chen,  Matthieu  Crussière, \IEEEmembership{Member, IEEE}  and Luc Le Magoarou, \IEEEmembership{Member, IEEE}

\thanks{Xuan Chen, Matthieu Crussière and Luc Le Magoarou are with IETR, INSA Rennes, France (e-mail: lastname.firstname@insa-rennes.fr) }

\thanks{This work is supported by the French national research agency (MoBAIWL project, grant ANR-23-CE25-0013)}

}

\markboth{IEEE WIRELESS COMMUNICATION LETTERS}
{Shell \MakeLowercase{\textit{et al.}}: Bare Demo of IEEEtran.cls for IEEE Journals}
\maketitle

\begin{abstract}
In this paper, we investigate the impact of hardware impairments in antenna arrays on the beamforming performance of multi-input multi-output (MIMO) communication systems. We consider two types of imperfections: per-element gain mismatches and inter-element spacing deviations.  We analytically determine the impairment configurations that result in the worst-case degradation. In addition,  the analytical expression of the average-case performance is also derived for comparison. Simulation and theoretical results show that the average SNR degradation remains relatively limited, whereas the worst-case scenarios can exhibit substantially higher losses. These findings provide clear insight into the robustness limits of MIMO systems under practical hardware imperfections.

\end{abstract}


\IEEEpeerreviewmaketitle

\section{Introduction}

\IEEEPARstart{A}{ntenna} arrays play a vital role in modern multi-input multi-output (MIMO) communication systems due to the spatial diversity and performance gains they provide \cite{11034774,9112566,9848357}. However, most theoretical models idealize the antenna array, assuming perfectly calibrated and precisely placed elements. While this simplifies analysis and design, it overlooks the fact that real-world arrays are subject to practical limitations that can affect their performance. As systems scale and rely more heavily on precise spatial processing, even small deviations from an ideal array can have non-negligible consequences \cite{yassine2022, chatelier2024}. This calls for a deeper investigation into the direct impact of hardware imperfections at the array level.

As with all physical-layer components, antenna arrays inevitably suffer from hardware impairments due to factors such as fabrication imperfections, temperature variations, or environmental conditions. Most existing studies focus on system-level impairments occurring after the antenna array, including I/Q imbalance, oscillator phase noise, and power amplifier non-linearity \cite{9306107,6891254,9294080,9786293}, typically assuming that the array itself is ideal. In contrast, hardware impairments at the antenna array level have been theoretically addressed only in a limited number of works. To the best of our knowledge, only specific experimental setups have studied this class of impairments \cite{10047354, chatelier2023, mateosramos2025, klaimi2025,10004742}. Among these, most works attempt to correct the impairments in practice through learning-based frameworks, however, none provide a theoretical analysis of their influence.

In this paper, we investigate the effect of antenna-level hardware impairments on the beamforming performance of MIMO systems. Specifically, we consider perturbations in the complex gain and the positions of antenna elements in a uniform linear array (ULA)\footnote{Note that although we focus on the ULA in this paper, the entire analysis and our following results can also be extended to a uniform planar array (UPA) or other types of array structures.}. Unlike previous work \cite{10047354}, which is limited to experimental observations, we provide a complete theoretical analysis of the beamforming degradation caused by these impairments. This enables a deeper and more general understanding of how such imperfections affect array performance in both average and worst-case scenarios. More precisely, given a certain tolerance on the position and gain impairments, we address two key questions: (i) What is the average beamforming loss? (ii) What is the worst-case beamforming loss, and how does the corresponding antenna array look like?

The rest of the paper is organized as follows: Section II presents both the signal model and the impairment model of the considered MU-MIMO communication scenario. Section III presents the theoretical analysis, including both the average case and the worst case analysis. Section IV provides simulation results to validate the results of Section III, while Section V concludes our work.

\textbf{Notations}: Boldface letters $\boldsymbol{a}$ and normal font letters ${a}$ represent vectors and scalars, respectively. In particular, $\boldsymbol{\vec{a}}$ denotes the conventional 3D vector defined in Euclidean space.  $\boldsymbol{a}^{*}$, $\boldsymbol{a}^T$ and $\boldsymbol{a}^H$ represent the complex conjugate, the transpose and the conjugate transpose of $\boldsymbol{a}$, respectively. $\mathbb{E}[.]$ represents the mathematical expectation function.


\section{System Model}


 We consider a base station (BS) equipped with a uniform linear array (ULA) of $N$ antennas, with a $\frac{\lambda}{2}$ spacing between adjacent elements. A user equipment (UE) is assumed to have a single antenna located at a distance from the BS with a direction of departure (DoD) angle $\theta$. Then, the position vector $\vec{\boldsymbol{a}}_i$ captures the theoretical/nominal position of the $i$th antenna along the $x$-axis with respect to the center of mass. We also define the unitary vector $\vec{\boldsymbol{u}}(\theta) = [\mathrm{cos}(\theta),\mathrm{sin}(\theta),0]^T$ which points in the direction of the DoD. For the sake of  conciseness,  this vector is simply denoted as $\vec{\boldsymbol{u}}$ in the rest of the paper.

\subsection{Signal Model}

The received signal $y$ at the UE side can be expressed as follows:
 \begin{equation}
 \label{system_model}
     y = \boldsymbol{h}^T\boldsymbol{w}s+n
 \end{equation}
 \noindent where $\boldsymbol{h}\in \mathbb{C}^N$ and $\boldsymbol{w} \in \mathbb{C}^N$ represent the channel vector and the precoding vector, respectively. $s$ is the normalized transmitted symbol with unitary power (\emph{i.e.} $\mathbb{E}\{| s |^2\}=1$) and  $n$ corresponds to the noise received for the UE ($n \sim \mathbb{C}\mathcal{N}(0,\sigma^2)$). The SNR can then be defined as:
 \begin{equation}
 \label{SNR_def}
     \mathrm{SNR} \triangleq \frac{\mathbb{E}\{|\boldsymbol{h}^T\boldsymbol{w}s|^2\}}{\mathbb{E}\{|n|^2\}} = \frac{|\boldsymbol{h}^T\boldsymbol{w}|^2}{\sigma^2}.
 \end{equation}
In our work, the precoder is given by:
\begin{equation}
 \label{vec_nom}
     \boldsymbol{w}=\boldsymbol{e}(\boldsymbol{\vec{u}}) = \frac{1}{\sqrt{N}}[e^{-j\frac{2\pi}{\lambda}\boldsymbol{\vec{a}_1}\cdot \boldsymbol{\vec{u}}},\ldots,e^{-j\frac{2\pi}{\lambda}\boldsymbol{\vec{a}_N}\cdot \boldsymbol{\vec{u}}}]^T, 
 \end{equation}
 \noindent which corresponds to the nominal steering vector for an ideal ULA  based on the UE DoD. The actual channel vector, however, incorporates antenna-level impairments and is expressed as:
 \begin{equation}
 \label{vec_pert}
     \boldsymbol{h} = \beta\boldsymbol{\tilde{e}}(\boldsymbol{\vec{u}}) = \beta[g_1 e^{-j\frac{2\pi}{\lambda}\boldsymbol{\tilde{\vec{a}}}_1\cdot \boldsymbol{\vec{u}}},\ldots,g_N e^{-j\frac{2\pi}{\lambda}\boldsymbol{\tilde{\vec{a}}}_N\cdot \boldsymbol{\vec{u}}} ]^T,
 \end{equation}
\noindent where $\beta$ corresponds to the channel attenuation. $\boldsymbol{\tilde{\vec{a}}}_i$ and $g_i$ ($\forall i \in \{1,\ldots, N\}$) respectively denote the actual  position of the $i$th antenna and its complex gain, both obtained from the impairment models described in the next section.

\subsection{Hardware Impairments Model}
\label{sec:model_pert}

In practice, the impairment can be due to various reasons (\emph{e.g.} fabrication default, wear etc.), so we model the impairment model as random variables, as in \cite{10047354}. Specifically, we assume that each antenna of the ULA suffers from independent impairment. In particular, the complex gain $g_i$ for the $i$th  antenna can be written as $g_i = \rho_i e^{j2\pi\varphi_i}$, where $\rho_i \sim \mathcal{U}(1-\delta_g, 1)$ with $ 0\le \delta_g \le 1$ and $\varphi_i \sim \mathcal{U}(-\alpha_g,\alpha_g)$ with $\alpha_g \ge 0$. Physically speaking, $\rho_i$
 may result from variations in gain from one element to another, and $\varphi_i$
 may be caused by differences in the transmission line lengths or by impedance mismatches among the antenna elements. Furthermore, we also consider that $\rho_i$ and $\varphi_i$ are independent for each antenna of the ULA. Similarly, the position impairment of each antenna can be modelled as a slight random shift of the nominal position\footnote{In the case of UPA, $\vec{\boldsymbol{n}}_{p,i} = [\varepsilon_{x,i},\varepsilon_{y,i},0]^T$.}: $\forall i \in\{1,\ldots,N\}$, $\boldsymbol{\tilde{\vec{a}}}_i = \boldsymbol{\vec{a}}_i + \lambda \boldsymbol{\vec{n}}_{p,i}$, where $\boldsymbol{\vec{n}}_{p,i} = [\varepsilon_{x,i},0,0]^T$, and $\varepsilon_{x,i} \sim \mathcal{U}(-\delta_p,\delta_p)$ with $\delta_p \ge 0$. Furthermore, as the impairment at the antenna level are in practice often small in amplitude, we further assume that the overall phase shift due to position and phase impairment are bounded, that is, $2\pi|\alpha_g+\delta_p|<\frac{\pi}{2}$, or equivalently $|\alpha_g+\delta_p|<\frac{1}{4}$.


\section{Theoretical Analysis}

In this section, we evaluate the influence of the impairment on the overall beamforming performance, in particular the SNR variation in average and in the worst case relative to the nominal case. From (\ref{SNR_def}) and (\ref{vec_nom}), the nominal case SNR without impairment, denoted by $\mathrm{SNR}_\mathrm{nom}$, can be defined as follows: 
    \begin{equation}
    \label{SNR_nom}
        \mathrm{SNR}_\mathrm{nom} \triangleq \frac{\mathbb{E}\{\left| \beta\boldsymbol{e}^H(\boldsymbol{\vec{u}})\boldsymbol{e}(\boldsymbol{\vec{u}})s |^2\right\}}{\mathbb{E}\{|n|^2\}} = \frac{N\beta^2}{\sigma^2} 
    \end{equation}
    \noindent with $\sigma^2$ the received noise power at the UE. Substituting the perturbed vector $\tilde{\boldsymbol{e}}(\vec{\boldsymbol{u}})$ defined in (\ref{vec_pert}) into (\ref{SNR_nom}), it is clear that the SNR in average or the worst case inevitably depends on the  correlation between the channel vector (\emph{i.e.} the perturbed vector $\boldsymbol{\tilde{e}}(\boldsymbol{\vec{u}})$) and the precoding vector (\emph{i.e.} the nominal vector $\boldsymbol{e}(\boldsymbol{\vec{u}})$), which can be expressed as follows:
\begin{equation}
\label{correl2}
    \left| \boldsymbol{\tilde{e}}^H(\boldsymbol{\vec{u}})\boldsymbol{e}(\boldsymbol{\vec{u}}) \right| = \frac{1}{\sqrt{N}}\left|\displaystyle\sum_{i=1}^{N} \rho_i e^{j2\pi(\varepsilon_{x,i}\mathrm{cos}(\theta)-\varphi_i)}\right|.
\end{equation}

\subsection{Worst Case Analysis}

\begin{figure}
\centerline{\includegraphics[width=0.65\columnwidth]{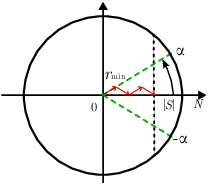}}
\caption{Graphical illustration of results of Lemma 1 with $N=4$.}
\label{fig:random_walk}
\end{figure}

In this section, we first identify the worst-case scenario associated with the impairment models described in Section \ref{sec:model_pert}. According to (\ref{correl2}), the worst-case corresponds to a  perturbation configuration that minimizes the correlation between $\boldsymbol{\tilde{e}}$ and $\boldsymbol{e}$. To determine such impairment configuration, we propose  applying the following lemma within our communication scenario:  

\begin{lemma}
\label{lemma1}
    Given $N$ vectors $\vec{\boldsymbol{v}_1},\ldots, \vec{\boldsymbol{v}_N}$ with $N$ even, where $\forall i \in \{1,\ldots,N\}$ $\vec{\boldsymbol{v}_i} = r_i e^{j\phi_i}$, $r_i \in [r_\mathrm{min}; r_\mathrm{max}]$ and $\phi_i \in [-\alpha;\alpha]$ with $|\alpha| < \frac{\pi}{4}$, then:
    \begin{equation}
        \mathrm{arg}\min_{r_i,\phi_i} |S| = \mathrm{arg}\min_{r_i,\phi_i} \left| \displaystyle\sum_{i=1}^{N} \vec{\boldsymbol{v}_i} \right| = \{r_\mathrm{min},p_i \alpha\},
    \end{equation}
\noindent where $p_i \in \{-1,1\}$ such that $\sum_{i=1}^{N}p_i = 0$.
\end{lemma}

\begin{remark}
 \normalfont{The result of Lemma \ref{lemma1} can intuitively be  understood by referring to Fig. \ref{fig:random_walk}. Specifically, minimizing the modulus of the sum of $N$ vectors is analogous to minimizing the endpoint distance of a random walk in the complex plane. Given that each step in the walk has a bounded magnitude and a direction constrained within specific limits, the minimal total displacement is achieved when each step has the smallest possible magnitude, and the phases are split evenly between the upper and lower bounds. That is, exactly $\frac{N}{2}$ steps take the upper-bound direction, and the remaining $\frac{N}{2}$ steps take the lower-bound direction. }
\end{remark}

\begin{proof}
     We start by considering $|S|^2$, since its minimum coincides with that of $|S|$. In particular, $|S|^2$ can be expanded using $S$ and $S^*$, as follows:
     \begin{align} \nonumber
         |S|^2  &= \left( \displaystyle \sum_{i=1}^{N}r_i e^{j\phi_i}\right)\left( \displaystyle \sum_{k=1}^{N}r_k e^{-j\phi_k}\right)\\ 
         & = \displaystyle \sum_{i=1}^{N} r_i^2 + 2\displaystyle \sum_{i=1}^{N}\displaystyle \sum_{j>i}r_i r_j \mathrm{cos}(\phi_i-\phi_j).
     \label{step1}
     \end{align}
     We first determine the value of each $r_i$. To do so, we calculate the derivative of $|S|^2$ with respect to a certain $r_k$ ($k \in \{1,\ldots,N\}$), as follows:
     \begin{equation}
         \label{derivative_S2}
         \frac{\partial |S|^2}{\partial r_k} = 2r_k + 2\displaystyle\sum_{i \neq k}r_k \mathrm{cos}(\phi_i-\phi_j).
     \end{equation}
     Since $\alpha < \frac{\pi}{4}$, it is guaranteed that the term $\mathrm{cos}(\phi_i-\phi_j)$ is always positive. Under this condition, it can be easily seen that the r.h.s. term of (\ref{derivative_S2}) is strictly positive. Thus, increasing any $r_i$ would always result in an increase in $|S|^2$. Consequently, minimizing $|S|^2$ requires first minimizing the contribution of each $r_i$, that is, setting $r_i = r_\mathrm{min}$. (\ref{step1}) then becomes:
     \begin{equation}
     \label{step2}
         |S|^2 \ge N r_\mathrm{min}^2 + 2r_\mathrm{min}^2\displaystyle\sum_{i=1}^{N}\displaystyle\sum_{j>i} \mathrm{cos}(\phi_i-\phi_j).
     \end{equation}
     Secondly, since $\mathrm{cos}(x)$ is concave on $[-2\alpha,2\alpha]$, the sum $\sum_{j>i}\mathrm{cos}(\phi_i-\phi_j)$ is concave on the convex set $[-\alpha,\alpha]$.  By Bauer's maximum principle \cite{Bauer1958MinimalstellenVF}, the minimum of $|S|^2$ is attained at an extreme point of the $N$-dimensional box $[-\alpha,\alpha]^N$, namely when each $\phi_i$ equals either $+\alpha$ or $-\alpha$.
     
     The remaining question is to determine the partition of $+\alpha$ and $-\alpha$ among all pairs of $\phi_i-\phi_j$. To do so, let's consider an integer $k \in \{0,\ldots,N\}$ and the two following sets:
     \begin{equation}
         \mathcal{A}^+ = \{i | \phi_i = \alpha\} \quad \mathrm{and} \quad \mathcal{A}^- =  \{i  |  \phi_i = -\alpha\}. \nonumber
     \end{equation}
     By definition, it is clear that $|\mathcal{A}^+| = k$ and $|\mathcal{A}^-| = N-k$. Using these two sets, we can further expand the cross-term of (\ref{step2}) as follows:
     \begin{align}\nonumber
         \displaystyle\sum_{i=1}^{N}\displaystyle\sum_{j>i}\mathrm{cos}(\phi_i-\phi_j) &= \displaystyle\sum_{i\in \mathcal{A}^+}\displaystyle\sum_{j\neq i, j\in \mathcal{A}^+} 1 + \displaystyle\sum_{i\in\mathcal{A}^-}\displaystyle\sum_{j\neq i, j\in\mathcal{A}^-} 1\\  &+\displaystyle\sum_{i\in\mathcal{A}^+}\displaystyle\sum_{j\in\mathcal{A}^-} \mathrm{cos}(2\alpha).
         \label{step3}
     \end{align}
     From (\ref{step3}), the cross-term has three contributions. (i) When all distinct pairs $(i,j)$ are in $\mathcal{A}^+$, which corresponds to $\binom{k}{2}$ terms. (ii) When all distinct pairs $(i,j)$ are in $\mathcal{A}^-$, which corresponds to $\binom{N-k}{2}$ terms. (iii) When $i$ or $j$ is in $\mathcal{A}^+$ or $\mathcal{A}^-$, which corresponds to $k(N-k)$ terms. Consequently, (\ref{step3}) can be further simplified as follows:
     \begin{align}\nonumber
         \displaystyle\sum_{i=1}^{N}\displaystyle\sum_{j>i}\mathrm{cos}(\phi_i&-\phi_j) = \binom{k}{2}+\binom{N-k}{2} \\ \nonumber
         & + k(N-k) \mathrm{cos}(2\alpha)\\
         = \Big(1-\mathrm{cos}(2\alpha)\Big)&(k^2-Nk) + \frac{N(N-1)}{2}.
         \label{step4}
     \end{align}
     Finally, since $\mathrm{cos}(2\alpha)$ is never less than  $-1$ under the given conditions, (\ref{step4}) is simply a convex quadratic polynomial in $k$, whose minimum is achieved at $k=\frac{N}{2}$. That is, exactly $\frac{N}{2}$ vectors take the direction $\alpha$ and the remaining $\frac{N}{2}$ vectors take the direction $-\alpha$, which concludes the proof. 
\end{proof}

The  direct application of Lemma \ref{lemma1} leads to the following Theorem 1, which summarizes the worst-case impairment configuration in our system:

\begin{theorem}
\label{theorem1}
For a ULA with $N$ antennas, under the hardware impairment model described in Section \ref{sec:model_pert} and for a given DoD $\theta$, The worst-case scenario that causes the greatest deterioration of beamforming performance occurs when:  $\forall i \in \{1,\ldots,N\}$
\begin{itemize}
    \item The amplitude of the gain, phase and position impairments are  maximal: $\rho_i = 1 - \delta_g$, $|\varphi_i| = \alpha_g $, $|\varepsilon_{x,i}| = \delta_p $
    \item The phase and position impairments are sign-aligned, so a positive deviation in one corresponds to a positive deviation in the other: $\varepsilon_{x,i} \mathrm{cos}(\theta) - \varphi_i = p_i (\delta_p+\alpha_g)$,
    \item The array is evenly split between positive and negative phase/position deviations: $p_i \in \{-1,1\}$ and $\sum_{i=1}^{N}p_i = 0$.
\end{itemize}
\end{theorem}

\begin{proof}
The proof of Theorem \ref{theorem1} follows directly by applying Lemma \ref{lemma1} to (\ref{correl2}), which gives $\vec{\boldsymbol{v}}_i = \rho_i e^{j2\pi(\varepsilon_{x,i}\mathrm{cos}(\theta)-\varphi_i)}$, and setting $r_\mathrm{min} = 1-\delta_g$,  $r_\mathrm{max} = 1$ and $\alpha = \alpha_g+\delta_p$.
\end{proof}

\begin{remark}
 \normalfont{Theorem 1 presents three conditions for the worst-case scenario in a ULA. In the absence of gain and phase impairments, these conditions can be simplified to two: the amplitude of the position impairment is maximal, with exactly half of the antennas experiencing a positive position shift and the other half experiencing a negative position shift. }
\end{remark}

\begin{figure}
\centerline{\includegraphics[width=0.8\columnwidth]{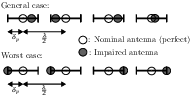}}
\caption{Illustration of the worst-case ULA configuration with $N=4$.}
\label{fig:worst_ULA}
\end{figure}

From Theorem \ref{theorem1}, a natural question arises: what does the worst-case antenna array look like? Fig. \ref{fig:worst_ULA} provides an example of a worst-case ULA configuration under the considered hardware impairments. In this example with $N=4$, the first two antennas are shifted by $\delta_p$, while the  remaining two are shifted by $-\delta_p$. Note that there are in total $6$ such configurations that result in the worst-case scenario (\emph{i.e.} $\binom{N}{N/2}$ combinations). Since the correlation involves the sum of each antenna’s contribution, the order of the impaired antennas does not affect the outcome. 

In the sequel, we denote the worst case channel vector by ${\boldsymbol{e}_\mathrm{wst}}$. From Theorem \ref{theorem1}, we can easily derive the following corollary regarding the worst case SNR loss: 

\begin{corollary}
\label{corol1}
    The worst case SNR variation, denoted by $\Delta \mathrm{SNR}$, for a ULA with N antennas and for a given DoD $\theta$, using the impairment model described in Section \ref{sec:model_pert}, is given as follows:
    \begin{equation}
    \label{key1}
        \Delta \mathrm{SNR}  = (1-\delta_g)^2 \mathrm{cos}^2\left(2\pi\big(\delta_p\mathrm{cos}(\theta)+\alpha_g\big)\right). 
    \end{equation}
\end{corollary}

\begin{proof}
    Using the results of Theorem \ref{theorem1}, we define the perturbed SNR under the worst-case hardware impairments, denoted by $\mathrm{SNR}_\mathrm{wst}$, as follows:
    \begin{align} \nonumber
        &\mathrm{SNR}_\mathrm{wst} \triangleq  \frac{\mathbb{E}\{|\beta\boldsymbol{e}_\mathrm{wst}^H(\boldsymbol{\vec{u}})\boldsymbol{e}(\boldsymbol{\vec{u}})s|^2\}}{\mathbb{E}\{|n|^2\}}\\ \nonumber
        &= \frac{\beta^2}{N\sigma^2}\left|\displaystyle\sum_{i=1}^{N}(1-\delta_g)e^{j2\pi p_i(\mathrm{cos}(\theta)\delta_p+\alpha_g)}\right|^2\\ \nonumber
        &= \frac{\beta^2(1-\delta_g)^2}{N\sigma^2}\left|N\mathrm{cos}(2\pi (\mathrm{cos}(\theta)\delta_p+\alpha_g))\right|^2\\
        &=\frac{N\beta^2}{\sigma^2}(1-\delta_g)^2\mathrm{cos}^2\left(2\pi (\mathrm{cos}(\theta)\delta_p+\alpha_g)\right).
    \label{SNR_worst1}
    \end{align}
    The SNR loss can then be simply defined as the ratio between $\mathrm{SNR}_\mathrm{wst}$ and $\mathrm{SNR}_\mathrm{nom}$, as follows:
    \begin{equation}
        \Delta\mathrm{SNR} \triangleq \frac{\mathrm{SNR}_\mathrm{wst}}{ \mathrm{SNR}_\mathrm{nom}} = (1-\delta_g)^2 \mathrm{cos}^2\left(2\pi\big(\delta_p\mathrm{cos}(\theta)+\alpha_g\big)\right),
    \end{equation}
    \noindent which concludes the proof.
\end{proof}

From the expression of the SNR loss given in Corollary \ref{corol1}, we first verify that $\Delta \mathrm{SNR} = 1$ when there is no impairment (\emph{i.e.} $\delta_g = \alpha_g = \delta_p = 0$), which is consistent with its physical interpretation. Furthermore, the DoD angle also plays a crucial role in the SNR loss. More precisely, (\ref{key1}) suggests that the worst-case SNR loss is an increasing function of $\theta$ for $\theta \in [0^{\circ};90^{\circ}]$. While this behavior is mathematically consistent, it also has a clear physical interpretation: as the DoD angle becomes more aligned with the broad side of the array, the beamforming process experiences less degradation since the antenna displacements are orthogonal to the wavefront in that extreme case. More strikingly, $\Delta \mathrm{SNR}$ does not depend on the array size, which may seem counterintuitive. This implies that the worst-case SNR loss remains the same for both small-scale and large-scale arrays. These trends will be further illustrated with simulation results in Section \ref{sec:simu}.

\subsection{Average Case Analysis}

In this section, we evaluate the average impact of both gain and position impairments on the SNR at the considered UE. To this end, we adopt the impairment models described in Section \ref{sec:model_pert}. In particular, we focus on quantifying the  SNR variation relative to the nominal case, whose expression is provided in the following theorem:

\begin{theorem}
\label{theorem2}
    The average SNR variation, denoted by $\Delta \overline{\mathrm{SNR}}$, for a ULA with N antennas and for a given DoD $\theta$, using the impairment model described in Section \ref{sec:model_pert}, is given as follows:
    \begin{align} \nonumber
        \Delta\mathrm{\overline{SNR}} =& \frac{3-3\delta_g+\delta_g^2}{3N}+\frac{(N-1)(4-4\delta_g+\delta_g^2)}{4N}\\
        &\times \mathrm{sinc}^2\left(2\pi \mathrm{cos}(\theta)\delta_p\right)\mathrm{sinc}^2\left(2\pi\alpha_g\right).
        \label{key2}
    \end{align}
\end{theorem}

\begin{proof}
    Similarly to the proof of Corollary \ref{corol1}, we start by calculating the average received signal power, using the independence between the different perturbation models and linearity of expectation, as follows:
    \begin{align} \nonumber
        \mathbb{E}\{|\beta\boldsymbol{\tilde{e}}^H(\boldsymbol{\vec{u}})\boldsymbol{e}&(\boldsymbol{\vec{u}})s|^2\} = \frac{\beta^2}{N}\displaystyle\sum_{i=1}^{N}\mathbb{E}\{\rho_i^2\} \\ +\frac{\beta^2}{N}\displaystyle\sum_{i=1}^{N}\displaystyle\sum_{k \neq i }\mathbb{E}\{\rho_i\}\mathbb{E}\{\rho_k\}&\mathbb{E}\{e^{j2\pi \left((\varepsilon_{x,i}-\varepsilon_{x,k})\mathrm{cos}(\theta)-(\varphi_i-\varphi_k)\right)}\}.
    \label{theo2_1}
    \end{align}
    Since each $\rho_i$ follows a uniform distribution, the following equality holds:
    \begin{align}
    \label{theo2_2}
        \mathbb{E}\{\rho_i^2\} = \mathrm{Var}\{\rho_i\} + \mathbb{E}\{\rho_i\}^2 = \frac{3-3\delta_g+\delta_g^2}{3}.
    \end{align}
    Consequently, (\ref{theo2_1}) can be simplified as follows:
    \begin{align}\nonumber
       & \mathbb{E}\{|\beta\boldsymbol{\tilde{e}}^H(\boldsymbol{\vec{u}})\boldsymbol{e}(\boldsymbol{\vec{u}})s|^2\} = \frac{\beta^2(3-3\delta_g+\delta_g^2)}{3} + \frac{4-4\delta_g+\delta_g^2}{4} \\
       &\times\frac{\beta^2}{N}\displaystyle\sum_{i=1}^{N}\displaystyle\sum_{k\neq i}\mathbb{E}\{e^{j2\pi \mathrm{cos}(\theta)(\varepsilon_{x,i}-\varepsilon_{x,k})}\}\mathbb{E}\{e^{j2\pi(\varphi_k-\varphi_i)}\}.
       \label{theo2_3}
    \end{align}
    Furthermore, it can be noticed that each term in the sum on the  {r.h.s.} of (\ref{theo2_3}) corresponds to the  product of the characteristic functions of the random variables $\varepsilon_{x,i}$ and $\varepsilon_{x,k}$ evaluated at $2\pi\mathrm{cos}(\theta)$ (or  $\varphi_k$ and $\varphi_i$ evaluated at $2\pi$, respectively), whose expressions are known and given by:
    \begin{align}\nonumber
            \mathbb{E}\{e^{j2\pi \mathrm{cos}(\theta)(\varepsilon_{x,i}-\varepsilon_{x,k})}\} &= \left(\frac{e^{j2\pi\mathrm{cos}(\theta)\delta_p}-e^{-j2\pi \mathrm{cos}(\theta)\delta_p}}{j4\pi \mathrm{cos}(\theta)\delta_p}\right)^2\\
            &= \mathrm{sinc}^2(2\pi \mathrm{cos}(\theta)\delta_p).
    \label{theo2_4}
    \end{align}
    Using (\ref{theo2_4}), (\ref{theo2_3}) can be further simplified as follows:
    \begin{align} \nonumber
         \mathbb{E}\{|\beta\boldsymbol{\tilde{e}}^H(&\boldsymbol{\vec{u}})\boldsymbol{e}(\boldsymbol{\vec{u}})s|^2\}  = \frac{\beta^2(3-3\delta_g+\delta_g^2)}{3} + \frac{4-4\delta_g+\delta_g^2}{4} \\
        & \times \beta^2(N-1)\mathrm{sinc}^2(2\pi \mathrm{cos}(\theta)\delta_p)\mathrm{sinc}^2(2\pi\alpha_g).
    \label{theo2_5}    
    \end{align}
    The average SNR variation can then be simply calculated by considering the ratio between the r.h.s. term of (\ref{theo2_5}) and the nominal received signal power $\mathbb{E}\{\left|\beta \boldsymbol{e}^H(\boldsymbol{\vec{u}})\boldsymbol{e}(\boldsymbol{\vec{u}})s |^2\right\}$, which equals $N\beta^2$ according to (\ref{SNR_nom}). This concludes the proof.  
\end{proof}

Similarly to the previously discussed worst-case SNR loss, it is straightforward to verify that $\Delta\overline{\mathrm{SNR}} = 1$  in the absence of hardware impairments. More interestingly, when position and phase impairments are absent (\emph{i.e.} $\delta_g = \alpha_g = 0$), the worst-case gain impairment (\emph{i.e.} $\delta_g = 1$) leads to an average SNR loss of  $\frac{3N+1}{12N}$, which corresponds to approximately a 6dB loss compared to the nominal case for large antenna arrays. Furthermore, for large arrays, $\Delta\overline{\mathrm{SNR}}$ tends to a constant value independent of $N$, an unexpected behavior that mirrors the trend observed in the worst-case SNR loss. Finally, it can be observed once again that the average SNR loss increases with the DoD angle $\theta$, a trend that will be confirmed through simulation in the next section.


\section{Simulation \& Discussion}
\label{sec:simu}



\begin{figure}
\centerline{\includegraphics[width=0.8\columnwidth]{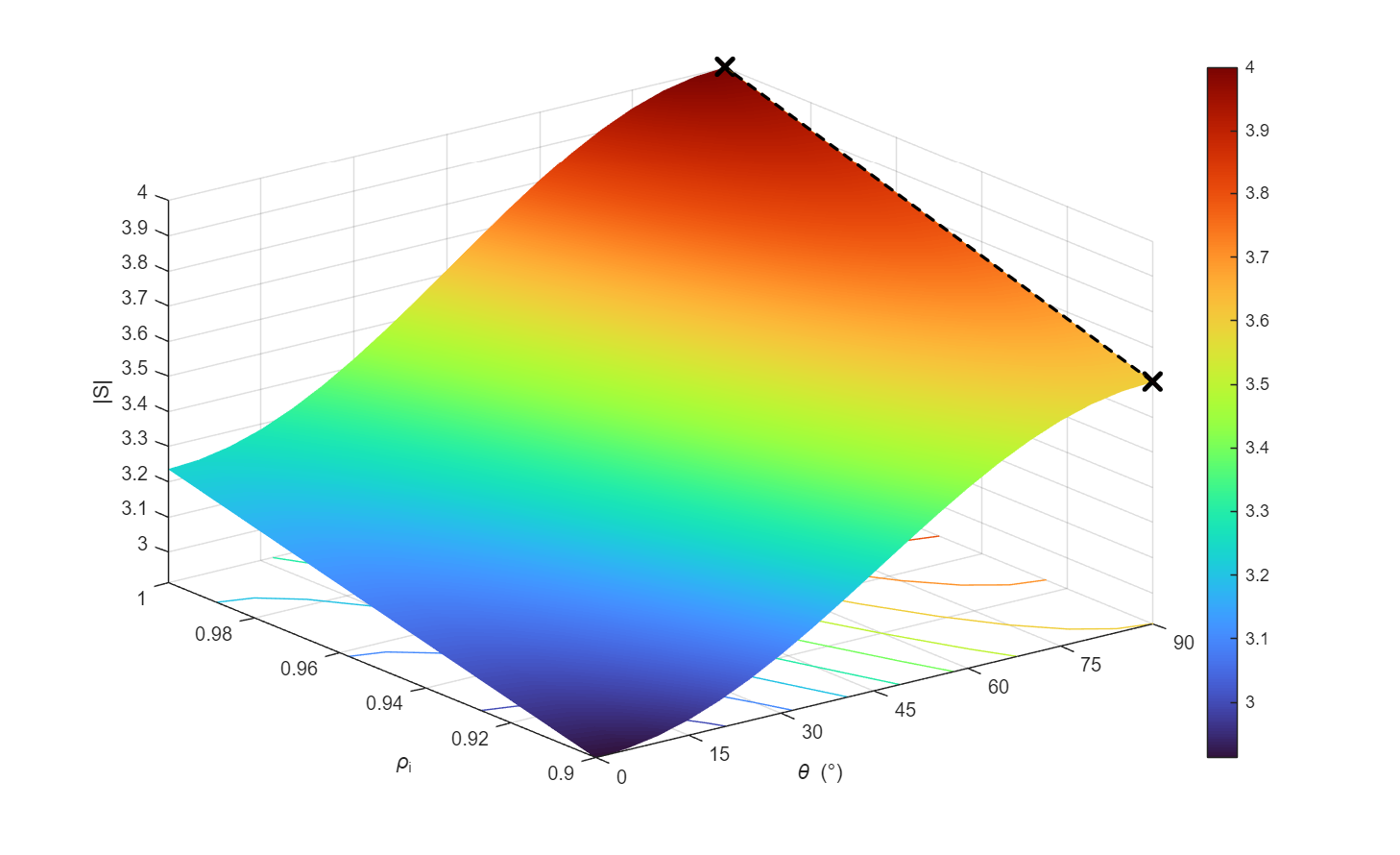}}
\caption{Worst-case correlation variation versus $\rho_i$ and DoD angle $\theta$ ($\delta_g = 0.1$, $\varphi_i = 0$, $\delta_p = 0.1$, $N=16$).}
\label{fig:S_var}
\end{figure}

In this section, the simulation results are presented and analyzed. The simulations are carried out using a ULA with $\frac{\lambda}{2}$ spacing. First, in Fig. \ref{fig:S_var}, we analyze the behavior of the correlation amplitude  $|S|$ under the worst-case scenario described in Theorem \ref{theorem1}. To clearly illustrate the impact of impairments, we vary both the per-antenna gain $\rho_i$ from $1-\delta_g$ to $1$, and the DoD angle $\theta$ from $0^{\circ}$ to $90^{\circ}$. For this analysis, the ULA size is fixed to $N=16$, and the position impairment is fixed to its worst-case configuration as defined in Theorem \ref{theorem1}, meaning that exactly half the antennas are set to  $\delta_p$ and the other half to $-\delta_p$. When $\theta = 90^{\circ}$, $|S|$ varies between $3.6$ and $4$, as highlighted by the markers and the dotted line in Fig. \ref{fig:S_var}.  The lower bound $3.6$ corresponds to the maximum gain degradation ($\rho_i = 0.9$ ), whereas the upper bound $4$ represents the nominal case with $\rho_i = 1$ (\emph{i.e.} an ideal line-of-sight condition without gain impairment). This behavior is consistent with the simulation setup, where $N=16$, and the nominal correlation amplitude equals $\frac{1}{\sqrt{N}}$ in the absence of impairments. Furthermore, it can be noted that $|S|$ decreases as $\rho_i$ is reduced, confirming the expected impact of gain perturbation. Additionally, $|S|$ increases with $\theta$, which aligns with previous theoretical results and explains why the SNR loss is also an increasing function of the DoD angle.

\begin{figure}
\centerline{\includegraphics[width=0.8\columnwidth]{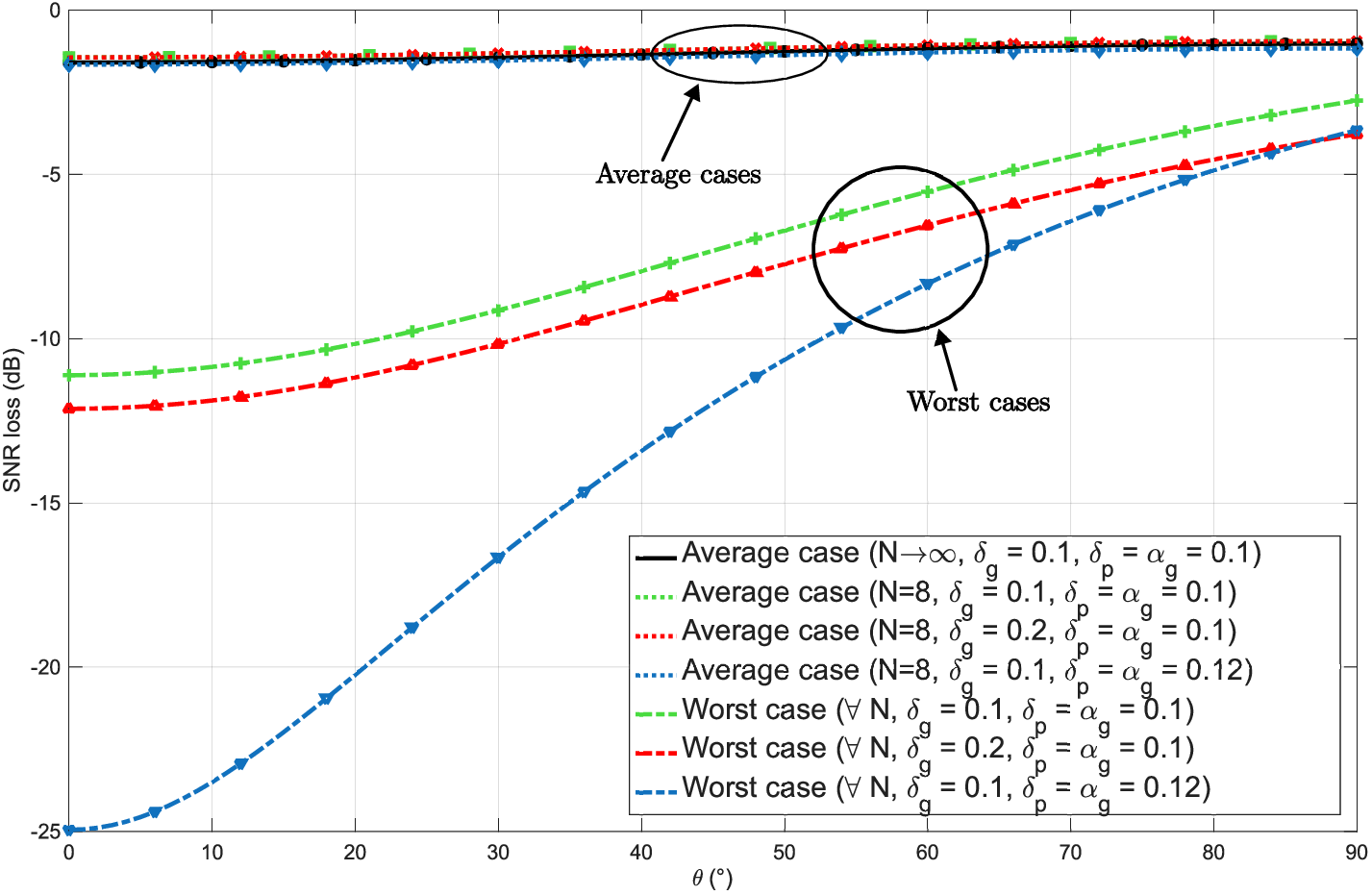}}
\caption{SNR loss vs $\theta$ for different array and perturbation configurations: lines represent theoretical plots, markers represent simulated plots.}
\label{fig:SNR_var}
\end{figure}

 Next, in Fig.\ref{fig:SNR_var}, we evaluate the SNR loss versus the steering angle $\theta$ for different array and perturbation configurations. In particular, the theoretical average and worst case results are obtained using (\ref{key1}) and (\ref{key2}) and are represented as plain lines in Fig. \ref{fig:SNR_var}, while the simulated results for each configuration are represented by markers in Fig. \ref{fig:SNR_var}. The simulated average-case result is obtained by aggregating 20000 realizations of the impairment variables, while the worst-case simulated result is obtained using the particle swarm optimization method to solve the optimization problem stated in Lemma \ref{lemma1}.

First, a perfect match between the theoretical and simulated results for all the considered configurations can be observed in Fig. \ref{fig:SNR_var}, which validates the theoretical expressions provided in Corollary \ref{corol1} and Theorem \ref{theorem2}. Furthermore, it is obvious that the SNR loss increases with $\theta$, which is consistent with our earlier mathematical results and also physically intuitive. 

Additionally, it can be observed that increasing the impairment parameters does not have a significant impact on the average SNR loss. This behavior can be explained by the fact that the average-case results vary only slightly for different values of $\delta_p$ and $\alpha_g$. According to Theorem 1 and the adopted amplitude error model, the parameters $\delta_g$, $\delta_p$, and $\alpha_g$ are all bounded. Consequently, the average SNR loss exhibits only minor variations across different impairment configurations. On the other hand, the same does not hold for the worst-case SNR loss. Specifically, doubling the gain impairment parameter $\delta_g$ only results in an additional SNR loss of about 1 dB compared to the $\delta_g = 0.1$ case. However, the phase and position impairment parameters, $\alpha_g$ and $\delta_p$ have a much more pronounced effect. When both $\alpha_g$ and $\delta_p$  are increased from 0.1 to 0.12, the worst-case SNR loss experiences a substantial degradation of up to -25dB at $\theta=0^{\circ}$. Such behavior can be understood by referring to Fig. \ref{fig:random_walk}: as the direction of each step in the random walk approaches the imaginary axis, the total resultant distance becomes much smaller, even if the step lengths remain nearly unchanged. In other words, the beamforming performance is considerably more sensitive to phase and position variations between elements than to gain variations. 

Finally, it can be observed that the SNR loss, in both the average and worst-case scenarios, depends only weakly on the array size $N$. As shown in (\ref{key1}), the worst-case SNR loss is theoretically independent of $N$, while the average-case results exhibit only minor variations as $N$ increases. This behavior can be explained by the fact that $\delta_g$ is bounded within $[0;1]$, hence, as $N$ grows, the influence of $\delta_g$ on $\Delta \mathrm{\overline{SNR}}$ rapidly diminishes, leading to an essentially constant average SNR loss for larger arrays.


\section{Conclusion}

In this paper, we investigated the beamforming degradation of a ULA in a MIMO communication system subject to hardware impairments. Unlike most existing works that address system-level effects, we focused on antenna-level perturbations. Our main contribution lies in the complete quantification of both average and worst-case beamforming losses under practical hardware impairments. We derived closed-form expressions for the SNR degradation and identified the impairment configuration leading to the worst-case loss. Interestingly, the worst-case beamforming loss was found to be independent of the array size, while the average loss only weakly depends on it and eventually becomes independent for large arrays. These findings provide useful insights into the robustness of MIMO arrays and serve as guidelines for array design and fabrication, helping to account for hardware tolerances and imperfections. In the future, the study could be generalized to take into account different array geometries (UPA for example) or different tasks such as DoD estimation.

\bibliographystyle{IEEEtran}
\bibliography{biblio}

\begin{thebibliography}{10}
\providecommand{\url}[1]{#1}
\csname url@samestyle\endcsname
\providecommand{\newblock}{\relax}
\providecommand{\bibinfo}[2]{#2}
\providecommand{\BIBentrySTDinterwordspacing}{\spaceskip=0pt\relax}
\providecommand{\BIBentryALTinterwordstretchfactor}{4}
\providecommand{\BIBentryALTinterwordspacing}{\spaceskip=\fontdimen2\font plus
\BIBentryALTinterwordstretchfactor\fontdimen3\font minus
  \fontdimen4\font\relax}
\providecommand{\BIBforeignlanguage}[2]{{%
\expandafter\ifx\csname l@#1\endcsname\relax
\typeout{** WARNING: IEEEtran.bst: No hyphenation pattern has been}%
\typeout{** loaded for the language `#1'. Using the pattern for}%
\typeout{** the default language instead.}%
\else
\language=\csname l@#1\endcsname
\fi
#2}}
\providecommand{\BIBdecl}{\relax}
\BIBdecl

\bibitem{11034774}
N.~Basha, K.Umadevi, Swetha.P, Sivaranjani.S, Yogeshwari.P, and Sridevikala.B,
  ``{Microstrip Patch Antenna Arrays for 5G IoT and Smart Communication
  Systems: A Literature Review},'' in \emph{2025 5th International Conference
  on Pervasive Computing and Social Networking (ICPCSN)}, 2025, pp. 494--499.

\bibitem{9112566}
M.~Sangwan, G.~Panda, and P.~Yadav, ``{A Literature Survey on Different MIMO
  Patch Antenna},'' in \emph{2020 International Conference on Inventive
  Computation Technologies (ICICT)}, 2020, pp. 912--918.

\bibitem{9848357}
S.~Nirmal and S.~Kumar, ``{Analysis of Diverse MIMO Antennas For Fifth
  Generation Application: A Review},'' in \emph{2022 IEEE Wireless Antenna and
  Microwave Symposium (WAMS)}, 2022, pp. 1--5.

\bibitem{yassine2022}
T.~Yassine and L.~Le~Magoarou, ``{mpNet: Variable Depth Unfolded Neural Network
  for Massive {MIMO} Channel Estimation},'' \emph{IEEE Transactions on Wireless
  Communications}, vol.~21, no.~7, pp. 5703--5714, 2022.

\bibitem{chatelier2024}
B.~Chatelier, J.~M. Mateos-Ramos, V.~Corlay, C.~H{\"a}ger, M.~Crussiere,
  H.~Wymeersch, and L.~L. Magoarou, ``{Physically parameterized differentiable
  MUSIC for DoA estimation with uncalibrated arrays},'' \emph{arXiv preprint
  arXiv:2411.15144}, 2024.

\bibitem{9306107}
W.~Belaoura, K.~Ghanem, M.~Z. Shakir, and K.~Qaraqe, ``{Impact of Hardware
  Impairments on the Performance of Millimeter-Wave Massive MU-MIMO Systems
  with Distributed Antennas},'' in \emph{2020 IEEE Eighth International
  Conference on Communications and Networking (ComNet)}, 2020, pp. 1--4.

\bibitem{6891254}
E.~Björnson, J.~Hoydis, M.~Kountouris, and M.~Debbah, ``{Massive MIMO Systems
  With Non-Ideal Hardware: Energy Efficiency, Estimation, and Capacity
  Limits},'' \emph{IEEE Transactions on Information Theory}, vol.~60, no.~11,
  pp. 7112--7139, 2014.

\bibitem{9294080}
Y.~Xu, H.~Xie, and R.~Q. Hu, ``{Max-Min Beamforming Design for Heterogeneous
  Networks With Hardware Impairments},'' \emph{IEEE Communications Letters},
  vol.~25, no.~4, pp. 1328--1332, 2021.

\bibitem{9786293}
W.~Belaoura, K.~Ghanem, M.~Nedil, and H.~Bousbia-Salah, ``{On the Impact of
  Hardware Impairments on mm-Wave MIMO Underground Channel Estimation},'' in
  \emph{2022 7th International Conference on Image and Signal Processing and
  their Applications (ISPA)}, 2022, pp. 1--5.

\bibitem{10047354}
S.~Prasad, M.~Meenakshi, and P.~H. Rao, ``{Hardware Impairments in mmWave
  Phased Arrays},'' in \emph{2022 IEEE Microwaves, Antennas, and Propagation
  Conference (MAPCON)}, 2022, pp. 891--896.

\bibitem{chatelier2023}
B.~Chatelier, L.~Le~Magoarou, and G.~Redieteab, ``{Efficient Deep Unfolding for
  {SISO}-{OFDM} Channel Estimation},'' in \emph{ICC 2023 - IEEE International
  Conference on Communications}, 2023, pp. 3450--3455.

\bibitem{mateosramos2025}
J.~M. Mateos-Ramos, C.~Häger, M.~F. Keskin, L.~Le~Magoarou, and H.~Wymeersch,
  ``{Model-Based End-to-End Learning for Multi-Target Integrated Sensing and
  Communication under Hardware Impairments},'' \emph{IEEE Transactions on
  Wireless Communications}, pp. 1--1, 2025.

\bibitem{klaimi2025}
N.~Klaimi, A.~Bedoui, C.~Elvira, P.~Mary, and L.~Le~Magoarou, ``{Model-based
  learning for joint channel estimationand hybrid MIMO precoding},''
  \emph{arXiv preprint arXiv:2505.04255}, 2025.

\bibitem{10004742}
S.~Rivetti, J.~Miguel Mateos-Ramos, Y.~Wu, J.~Song, M.~F. Keskin,
  V.~Yajnanarayana, C.~Häger, and H.~Wymeersch, ``{Spatial Signal Design for
  Positioning via End-to-End Learning},'' \emph{IEEE Wireless Communications
  Letters}, vol.~12, no.~3, pp. 525--529, 2023.

\bibitem{Bauer1958MinimalstellenVF}
\BIBentryALTinterwordspacing
H.~Bauer, ``{Minimalstellen von Funktionen und Extremalpunkte},'' \emph{Archiv
  der Mathematik}, vol.~9, pp. 389--393, 1958. [Online]. Available:
  \url{https://api.semanticscholar.org/CorpusID:120811485}
\BIBentrySTDinterwordspacing

\end{thebibliography}
\end{document}